\documentclass[11pt,a4paper]{article}
\usepackage[T1]{fontenc}
\usepackage[utf8]{inputenc}
\usepackage{lmodern,microtype}
\usepackage[margin=23mm]{geometry}
\usepackage{amsmath,amssymb,amsthm,booktabs,tabularx,array}
\usepackage{xcolor,enumitem,fancyhdr}
\usepackage{xurl}
\usepackage[colorlinks=true,linkcolor=blue!45!black,urlcolor=blue!45!black,citecolor=blue!45!black]{hyperref}
\hypersetup{pdftitle={Exact Locality Gaps for Matchable Semi-Matchings},pdfauthor={Marek Gałązka and Hanna Wdowicka}}
\setlist{itemsep=2pt,topsep=4pt}
\newtheorem{theorem}{Theorem}
\newtheorem{lemma}[theorem]{Lemma}
\newtheorem{corollary}[theorem]{Corollary}
\newcommand{\OPT}{\operatorname{OPT}}
\newcommand{\gap}{\Gamma^{\mathrm{match}}}
\newcommand{\eps}{\varepsilon}

\begin{document}
\begin{center}
{\LARGE\bfseries Exact Locality Gaps for\\[4pt] Matchable Semi-Matchings}\\[10pt]
{\large Marek Gałązka\textsuperscript{a,*}\qquad Hanna Wdowicka\textsuperscript{b}}\\[8pt]
{\small\textsuperscript{a}Faculty of Mathematics and Computer Science,\\
Adam Mickiewicz University, Uniwersytetu Poznańskiego 4,\\
61-614 Poznań, Poland}\\[5pt]
{\small\textsuperscript{b}Department of Statistics,\\
Poznań University of Economics and Business,\\
Al. Niepodległości 10, 61-875 Poznań, Poland}\\[6pt]
{\small\textsuperscript{*}Corresponding author:
\href{mailto:galazka@amu.edu.pl}{galazka@amu.edu.pl}}\\[2pt]
{\small\href{mailto:hanna.wdowicka@ue.poznan.pl}{hanna.wdowicka@ue.poznan.pl}}
\end{center}

\begin{abstract}
An assignment of tasks to servers can resist every small improvement and
still make tasks wait longer than necessary. We determine exactly how
inefficient such an assignment can be when each task requires one unit of
service and the eligibility constraints permit all tasks to use distinct
servers. For every move size $r$ and maximum current server load $K$, we
give a closed formula for the worst ratio between locally optimal and
globally optimal total completion time. Local optimality here allows every
feasible reassignment changing at most $r$ tasks. Every finite-cap bound is
attained on a tree where each task has at most two eligible servers. Thus
the worst behavior already occurs under simple eligibility constraints.
At load cap two, the exact ratio is $1+1/(r+2)$, attained on a path with
$r+2$ tasks. Without a load cap, the worst-case supremum is $3/2$ for
single-task moves and approximately $1.294503159$ for two-task moves; its
excess above one is $1/(r+2)+O(2^{-r}/r)$ as $r$ grows. The proof uses an
explicit rational potential on a comparison graph and matching extremal
constructions. These results give sharp guarantees for bounded-size local
search on matchable semi-matchings, including exact guarantees under
degree bounds.
\end{abstract}


\section{Introduction}

Suppose that a collection of tasks must be assigned to servers, but each
task can use only certain servers. Every task takes one unit of time, and
each server processes its assigned tasks one after another. Three tasks
sent to one server finish at times one, two and three, with total completion
time six. If they can instead use three different servers, their total
completion time is three. Eligibility restrictions therefore create a
basic optimization problem: distribute the tasks so that they finish as
early as possible in total.

A natural local-search procedure repeatedly changes the servers of a few
tasks whenever this lowers the total cost. It can stop even though a larger
coordinated change would help. We calculate the exact largest loss that
can remain at such a stopping point, assuming that the eligibility
restrictions permit all tasks to use distinct servers. The answer depends
on how many tasks may move together and on how many tasks currently share
one server. Our worst examples are trees with at most two choices per
task, and even a path exhibits the leading loss. Thus this obstruction to
local improvement is present in very simple assignment problems.

For a small example, take three tasks and three servers. The first task
can use only the first server; the second can use the first or second;
the third can use the second or third. Assign the first two tasks to the
first server and the remaining task to the second. The total completion
time is four. Moving either eligible task alone leaves that cost unchanged,
but moving the second and third tasks together places one task on each
server and lowers the cost to three. This example attains our exact
$4/3$ bound for single-task moves when no server has more than two tasks.

The graph formulation of this problem is a \emph{semi-matching}: tasks,
which we call clients from now on, form one side of a bipartite graph,
servers form the other, and edges specify eligibility. An assignment has
cost $\sum_v a_v(a_v+1)/2$, where $a_v$ is the number of clients assigned
to server $v$. An instance is \emph{matchable} when a matching covers every
client, so an assignment with no shared server is feasible. The
\emph{locality gap} is the worst ratio of the cost of a locally optimal
assignment to the global optimum. Semi-matchings and their improvement
paths were developed in \cite{harvey,czygrinow}.

The usual alternating-path move transfers one unit of load from a starting
server to an ending server through a chain of client reassignments,
preserving all intermediate loads. The move improves the cost precisely when the starting
load exceeds the ending load by at least two. We prove that excluding
improving paths moving at most $r$ clients is equivalent to excluding
\emph{every} improving reassignment of at most $r$ clients. This makes the
path-based guarantee a guarantee for the full bounded-size neighborhood.
Czygrinow et al.\ \cite{czygrinow} analyze short improving paths in a
distributed approximation framework. Their Theorem 2 gives $1+1/r$ on
matchable instances. We replace that upper bound by the exact worst-case
ratio on this class, with a matching example for every finite load cap.

\subsection{Our results and proof method}

Our main theorem gives a closed formula for the worst locality gap as a
function of the radius $r$ and maximum assigned load $K$. Every finite-cap
value is attained by a finite bipartite tree with client degree at most two
and server degree at most $K$. Consequently, the same formula gives the
exact gap under a maximum-degree bound. At $K=2$, a path with $r+2$ clients
attains $1+1/(r+2)$. Without a load cap, the exact supremum is the limit of
the finite-cap formula; it is not attained by a finite instance and has
excess $1/(r+2)+O(2^{-r}/r)$ above one. We also show that excluding short
improving paths is equivalent to local optimality under all feasible
reassignments changing at most $r$ clients.

The upper bound compares a given assignment with a client-covering
matching. This produces a directed graph with one incoming arc per server
and outdegree equal to the current load. A state records the distance to a
recent ancestor whose load is one higher. We assign an explicit rational
potential to each state and prove inequalities that telescope over the
comparison graph. A layered tree construction makes the resulting bound
tight. The argument is finite for every $r,K$ and includes directed cycles
and self-loops in the comparison graph.

\subsection{Related work}

Caragiannis et al.\ \cite{caragiannis} study the total latency of Nash
equilibria and greedy online assignments in load-balancing games with
linear latency functions. For identical latency $\ell(a)=a$, the social
cost is $Q(M)=\sum_v a_v^2$, whereas our objective satisfies
$C(M)=(Q(M)+n)/2$. Their affine-latency model also includes
$\ell(a)=(a+1)/2$, whose total latency is exactly $C(M)$ and whose unilateral
stability condition agrees with $1$-non-swappability. Thus the objective
alone does not distinguish the present setting. Our theorem determines
the exact gap as a function of both $r$ and $K$ under the additional
assumption of a client-covering matching.

Potential certificates and optimization-based bounds also appear in the
congestion-game literature. The full version of Ravindran Vijayalakshmi
and Skopalik~\cite{sv} gives a related one-step certificate in Lemma 4 and
a dual-based realization by singleton choices in Lemma 28.
Paccagnan et al.\ \cite{paccagnan} develop optimization-based certificates
for congestion games. Our contribution is the explicit radius-dependent
certificate and matching extremal family for the class considered here.

Correa and Mu\~noz~\cite{correa} study single-job jumps and pairwise
interchanges for minsum scheduling. Their unit-job example in Section 5.4
has ratio $311/204$, and they report larger examples reaching $1.533$.
Its optimum includes load-two machines, placing it outside the matchable
class. Pairwise interchanges and alternating chains also define different
neighborhoods. Kashaev~\cite{kashaev}, Section 6, develops dual-fitting
analyses of single-job local search and a surrogate-potential algorithm.
These works provide context for the one-move case; the parameters here
describe simultaneous reassignments of up to $r$ clients and a cap on the
current load.

Distributed token-dropping methods give another algorithmic context for
stable assignments~\cite{brandt}. We analyze the quality of an assignment
once the radius condition holds. The matchability assumption is central
to the comparison graph; instances whose optimum requires a load above
one are outside the scope of the theorem.

\section{Model and exact statement}

Let $G=(U,V,E)$ be a finite bipartite graph, where $U$ is the client set,
$V$ is the server set, and $n=|U|>0$. A semi-matching $M$ assigns every
client to one adjacent server. Write $a_v$ for its load at server $v$, and put
\[
 c_a=\frac{a(a+1)}2,\qquad C(M)=\sum_{v\in V}c_{a_v}.
\]
The optimum $\OPT(G)$ minimizes this cost over all assignments. We call
the instance \emph{matchable} if it admits a matching covering $U$.
Throughout the paper we assume matchability. Consequently
$\OPT(G)=n$: the matching has cost $n$, and $c_a\geq a$ for every load.
The numerical condition $n\leq|V|$ alone does not imply matchability.

An alternating path moving $t$ clients has $2t$ bipartite edges. All
intermediate server loads remain unchanged; the starting load $a$ loses
one and the ending load $b$ gains one. The cost reduction is
\[
 c_a+c_b-c_{a-1}-c_{b+1}=a-b-1.
\]
We say that $M$ is \emph{$r$-non-swappable} if there is no such improving
path with $t\leq r$. This agrees with the path-length convention of
Czygrinow et al.\ \cite{czygrinow}. Their Theorem 2 yields $1+1/r$ on the
matchable class. The radius $r$ bounds the number of moved clients, not
the number of bipartite edges.

Define $\gap_{r,K}$ as the supremum of $C(M)/\OPT(G)$ over matchable
instances and $r$-non-swappable assignments with $\max_v a_v\leq K$.
Define $\gap_r$ by omitting this load cap. Here $K$ bounds the
\emph{current assigned load}; it is not a graph-degree bound.

\begin{lemma}[Equivalence with bounded client exchanges]\label{lem:exchange}
An assignment is $r$-non-swappable if and only if no feasible reassignment
changing the assigned server of at most $r$ clients has smaller total cost.
This equivalence holds without the matchability assumption.
\end{lemma}
\begin{proof}
An improving alternating path moving at most $r$ clients is itself such a
reassignment. Conversely, suppose a feasible assignment $N$ changes
$t\leq r$ clients and has smaller cost. For each changed client, draw an
arc from its server under $M$ to its server under $N$. Write $a_v,b_v$ for
the respective loads. The resulting directed multigraph has $t$ arcs and
excess outdegree $a_v-b_v$ at $v$.

Decompose its arcs into directed cycles and paths from net donors
$x$, where $a_x>b_x$, to net receivers $y$, where $b_y>a_y$.
This is the usual integral flow decomposition: add a source supplying
$a_x-b_x$ units at each donor and a sink receiving $b_y-a_y$ units at each
receiver, then remove source--sink paths and residual cycles. Cycles do not
change any load. Each donor--receiver path pairs one unit removed from its
donor with one unit added to its receiver. A removed unit saves at most
$a_x$, while an added unit costs at least $a_y+1$. Consequently
\[
 C(M)-C(N)\leq\sum_{\text{donor--receiver paths }x\leadsto y}
 (a_x-a_y-1).
\]
Since the left side is positive, some path has $a_x\geq a_y+2$.
Remove any repeated-vertex segments from that path. The remaining simple
server path has at most $t$ arcs, uses distinct changed clients, and lifts
to an improving alternating path in the original bipartite graph. This
contradicts $r$-non-swappability.
\end{proof}

Equivalently, this is unit-time, unit-weight scheduling on restricted
identical machines with total completion time as objective and an optimum
of makespan one. In congestion-game terminology the objective is Rosenthal's
potential for identical latency functions $\ell(a)=a$, rather than the usual
sum of player latencies. If $Q(M)=\sum_v a_v^2$, then $Q(M)=2C(M)-n$;
on the matchable class the corresponding exact squared-load gap is
$2\gap_{r,K}-1$.

\begin{theorem}[Exact locality gaps]\label{thm:main}
For every integer $r\geq1$, $\gap_{r,1}=1$. For $K\geq2$, let
\begin{align}
 E_{r,K}&=\sum_{a=2}^{K-1}\frac1{(a!)^r}, &
 T_{r,K}&=\sum_{a=2}^{K-1}\frac1{a(a-1)(a!)^r},\label{eq:sums}\\
 Z_{r,K}&=r+2-T_{r,K}.\nonumber
\end{align}
Empty sums are zero. Then
\begin{equation}\label{eq:finite}
 \boxed{\gap_{r,K}=1+\frac{2+E_{r,K}}{2Z_{r,K}}.}
\end{equation}
For every $r\geq1$, $K\geq2$, this value is attained by a finite instance
whose eligibility graph is a tree, with client degree at most two and
server degree at most $K$.

Writing $E_r=\sum_{a=2}^{\infty}(a!)^{-r}$ and
$T_r=\sum_{a=2}^{\infty}[a(a-1)(a!)^r]^{-1}$, we have
\begin{equation}\label{eq:infinite}
 \boxed{\gap_r=\beta_r:=1+\frac{2+E_r}{2(r+2-T_r)}.}
\end{equation}
This unrestricted-load supremum is approached by the finite tree constructions
as $K\to\infty$ and is not attained by any finite instance. Thus restricting
the finite-cap or unrestricted supremum to matchable tree instances does
not change its value.
\end{theorem}

The rest of the proof is finite for every $r,K$: there is no passage to
an infinite-dimensional optimization problem. We first prove the upper
bound, then construct an instance attaining it. Equation~\eqref{eq:infinite}
follows by taking suprema over the load cap.

\section{One-parent comparison graphs}

\begin{lemma}\label{lem:matching}
There exists a matching $B$ covering all clients and all servers that
have positive load under $M$.
\end{lemma}
\begin{proof}
For every positive-load server, select one client assigned to it by $M$.
These edges form a matching because each client has only one assignment.
Augment this matching until it is maximum. Augmentation preserves coverage
of every previously covered server. Since $G$ admits a client-covering
matching, the resulting maximum matching covers all clients as well.
\end{proof}

Fix such a $B$ and retain the set $W$ of its $n$ covered servers.
No server outside $W$ contributes positive load or cost under $M$.
For every client $u$, draw an arc from its $M$-server to its $B$-server.
\emph{Keep self-loops.} The resulting directed graph $D$ satisfies
\begin{equation}\label{eq:degrees}
 \deg_D^-(v)=1,\qquad \deg_D^+(v)=a_v\quad(v\in W).
\end{equation}
Thus every server has a unique parent, including vertices on directed
cycles. Self-loops are useful for preserving~\eqref{eq:degrees}.

\begin{lemma}\label{lem:walks}
Every directed walk of length at most $r$ in $D$, from $v$ to $w$,
satisfies $a_v\leq a_w+1$.
\end{lemma}
\begin{proof}
If the endpoint inequality failed, the endpoints would be distinct.
Delete repeated-vertex segments and self-loops to obtain a simple directed
path of no greater length with the same endpoints. Its arcs correspond
to distinct clients and to an alternating path in $G$: each non-loop arc
uses a current assignment followed by a different $B$-assignment.
The endpoint loads differ by at least two, contradicting
$r$-non-swappability.
\end{proof}

For a vertex of load $a$, examine its unique ancestors at distances
$1,\ldots,r-1$. Each has load at most $a+1$ by Lemma~\ref{lem:walks}.
If none has load $a+1$, give the vertex state $(a,0)$. Otherwise, if the
nearest such ancestor is at distance $d$, give it state
\begin{equation}\label{eq:state}
 (a,j),\qquad j=r-d\in\{1,\ldots,r-1\}.
\end{equation}
When $r=1$, the ancestor window is empty and every state has second
coordinate zero. A load-zero vertex always has state $(0,r-1)$:
its parent must have positive outdegree, hence load one by
Lemma~\ref{lem:walks}. For $r=1$, this simply means $(0,0)$.

\begin{lemma}[Allowed transitions]\label{lem:transitions}
Every arc of $D$ follows one of these transitions:
\begin{align}
 (a,0)&\longrightarrow(a-1,r-1)
       &&\text{or }(b,0)\text{ with }b\geq a,\label{eq:free}\\
 (a,j)&\longrightarrow(a,j-1)
       &&\text{or }(b,0)\text{ with }b>a,\quad j\geq1.\label{eq:wait}
\end{align}
States with load zero have no outgoing arcs.
\end{lemma}
\begin{proof}
For an arc from load $a$ to load $b$, Lemma~\ref{lem:walks} gives
$b\geq a-1$. If $j>0$, an ancestor of load $a+1$ occurs at distance
$d=r-j\leq r-1$. A drop to $a-1$ would then produce a walk of length
$d+1\leq r$ with load drop two. Therefore a drop is possible only from
state $(a,0)$; its child's nearest higher-load ancestor is its parent,
giving state $(a-1,r-1)$.

If $b=a$, the nearest higher-load ancestor recedes by one step, so
the second coordinate becomes $\max\{j-1,0\}$. If $b>a$, all ancestors
in the child's window have load at most $a+1\leq b$: the parent has load
$a$, and the other relevant ancestors satisfy Lemma~\ref{lem:walks}.
The child's second coordinate is therefore zero. This reasoning uses
ancestor walks, so it also applies on cycles and at self-loops.
\end{proof}

The state is a compact record of a recent higher load. It does not
assume that loads decrease monotonically along all directed paths.
Keeping upward transitions in~\eqref{eq:free}--\eqref{eq:wait} is essential.

\section{A finite rational potential certificate}

\begin{lemma}[Certificate criterion]\label{lem:certificate}
Fix $r,K$ and a number $R$. Suppose a potential $P_{a,j}$ is defined
for $(0,r-1)$ and all $(a,j)$ with $1\leq a\leq K$, $0\leq j<r$.
Suppose $P_{0,r-1}\geq-R$ and, for every allowed transition
$(a,j)\to(b,k)$ with $a>0$,
\begin{equation}\label{eq:certificate}
 P_{a,j}\geq c_a-R+aP_{b,k}.
\end{equation}
Then every assignment in the definition of $\gap_{r,K}$ satisfies
$C(M)\leq R\OPT(G)$.
\end{lemma}
\begin{proof}
Use the states of Lemma~\ref{lem:transitions}. For a positive-load
vertex $v$, apply~\eqref{eq:certificate} to its $a_v$ outgoing arcs,
divide each inequality by $a_v$, and sum. For a zero-load vertex,
use the separate leaf inequality. In both cases,
\[
 P(v)\geq c_{a_v}-R+\sum_{v\to w}P(w).
\]
Summing over $W$ cancels all potentials, because every vertex occurs
once as an arc target. Hence $C(M)\leq R|W|=Rn=R\OPT(G)$.
The potentials may be negative; no sign condition on them is required.
\end{proof}

For the rest of this section, fix $r\geq1$, $K\geq2$, and set
\begin{equation}\label{eq:epsilon}
 \eps=\frac{2+E_{r,K}}{2(r+2-T_{r,K})},\qquad R=1+\eps,
 \qquad g_a=c_a-R.
\end{equation}
Since $T_{r,K}<\sum_{a=2}^{\infty}1/[a(a-1)]=1$, the denominator
is positive. We shall use
\begin{equation}\label{eq:epslower}
 \eps\geq\frac1{r+2}.
\end{equation}

Define the low-load potentials by
\begin{equation}\label{eq:low}
 P_{0,r-1}=-R,\qquad P_{1,j}=-1-(j+2)\eps\quad(0\leq j<r).
\end{equation}
For $2\leq a\leq K$, put
\begin{align}
 s_a&=-\frac a2-1+\frac{\eps}{a-1}
       =\frac{R-c_a}{a-1},\label{eq:sa}\\
 L_a&=a!((a-1)!)^{r-1}=a((a-1)!)^r,\nonumber\\
 \delta_a&=L_a\sum_{k=a}^{K-1}
      \frac{\tfrac12+\eps/[k(k-1)]}{(k!)^r},\label{eq:delta}\\
 P_{a,j}&=s_a-a^j\delta_a\quad(0\leq j<r).\label{eq:potential}
\end{align}
In particular, $\delta_a\geq0$ and $\delta_K=0$. All these numbers
are rational and involve only finite sums.

\subsection{Identities and monotonicity}

\begin{lemma}\label{lem:identities}
The potential satisfies the following identities and inequalities:
\begin{align}
 P_{a,j}&=g_a+aP_{a,j-1}
       &&(1\leq a\leq K,\ 1\leq j<r),\label{eq:same}\\
 P_{a,0}&=g_a+aP_{a-1,r-1}
       &&(1\leq a\leq K),\label{eq:drop}\\
 P_{a,0}&\geq g_a+aP_{a,0}
       &&(1\leq a\leq K).\label{eq:self}
\end{align}
For each positive $a$, $P_{a,j}$ is nonincreasing in $j$.
For $r\geq2$, it also satisfies
\begin{equation}\label{eq:interlace}
 P_{a,r-2}\geq P_{a+1,0}\qquad(1\leq a<K).
\end{equation}
\end{lemma}
\begin{proof}
For $a\geq2$, $g_a=-(a-1)s_a$, proving~\eqref{eq:same} directly
from~\eqref{eq:potential}. For $a=1$, use $g_1=-\eps$ and
\eqref{eq:low}. Monotonicity in $j$ follows from $\delta_a\geq0$
and from the explicit formula at load one.

To prove~\eqref{eq:drop}, the case $a=1$ is immediate from
\eqref{eq:low}. For $a=2$, the definition of $\delta_2$ gives
\[
 \delta_2=E_{r,K}+2\eps T_{r,K}.
\]
Equation~\eqref{eq:epsilon} therefore implies
\[
 P_{2,0}=-2+\eps-E_{r,K}-2\eps T_{r,K}
        =-(2r+3)\eps
        =g_2+2P_{1,r-1}.
\]
For $a\geq3$, separating the first summand in $\delta_{a-1}$ yields
\begin{equation}\label{eq:delta-rec}
 \delta_a=a(a-1)^{r-1}\delta_{a-1}
       -a\left(\frac12+\frac{\eps}{(a-1)(a-2)}\right).
\end{equation}
Combine this with
$s_a-s_{a-1}=-1/2-\eps/[(a-1)(a-2)]$ and
$g_a=-(a-1)s_a$ to obtain~\eqref{eq:drop}.

The slack in~\eqref{eq:self} is $\eps$ when $a=1$ and
$(a-1)\delta_a$ when $a\geq2$, so it is nonnegative.

It remains to prove~\eqref{eq:interlace}. At $a=1$,
\[
 P_{1,r-2}-P_{2,0}=(r+3)\eps-1
       \geq\frac1{r+2}>0.
\]
For $2\leq a<K$, write $t_a=1/2+\eps/[a(a-1)]$. The first
summand of~\eqref{eq:delta} gives
\[
 \delta_a\geq a^{1-r}t_a,
\]
while~\eqref{eq:delta-rec} at $a+1$ gives
$\delta_{a+1}=(a+1)a^{r-1}\delta_a-(a+1)t_a$. Consequently,
\begin{align*}
 P_{a,r-2}-P_{a+1,0}
 &=a^{r-2}\bigl(a(a+1)-1\bigr)\delta_a-a t_a\\
 &\geq\left(1-\frac1a\right)t_a\geq0.
\end{align*}
\end{proof}

\subsection{Verification of every transition}

Assume first that $r\geq2$. Write $p_a=P_{a,0}$ for positive loads.
By monotonicity in $j$ and~\eqref{eq:interlace},
$p_a\geq P_{a,r-2}\geq p_{a+1}$. Thus $p_a$ is nonincreasing in $a$.

For a transition from $(a,0)$ to $(a-1,r-1)$, the certificate is
the equality~\eqref{eq:drop}. For a transition from $(a,0)$ to
$(b,0)$ with $b\geq a$, use $p_b\leq p_a$ and~\eqref{eq:self}.
For a transition from $(a,j)$ to $(a,j-1)$, use~\eqref{eq:same}.
Finally, if $j\geq1$ and $b>a$, then
\[
 P_{b,0}\leq P_{a+1,0}\leq P_{a,r-2}\leq P_{a,j-1},
\]
because $j-1\leq r-2$. Equation~\eqref{eq:same} proves the
certificate for this last type of transition as well. There are no
upward transitions beyond $K$, and the zero-load inequality is equality.

When $r=1$, all states have second coordinate zero. Put $p_0=P_{0,0}$.
Combining~\eqref{eq:drop} and~\eqref{eq:self} gives
$p_{a-1}\geq p_a$ for $1\leq a\leq K$. Every allowed target has
load $b\geq a-1$, so
\[
 g_a+a p_b\leq g_a+a p_{a-1}=p_a.
\]
The criterion therefore holds for $r=1$ too, without using the
interlacing inequality with an invalid index $r-2$.

Lemma~\ref{lem:certificate} proves the upper bound in
\eqref{eq:finite}. When $K=1$, every assigned client contributes exactly
one to the cost, giving $\gap_{r,1}=1$ separately.

\section{Finite constructions attaining the bound}

Fix $r\geq1$, $K\geq2$, and put $H=r(K-1)+1$.
There are layers $i=0,1,\ldots,H$, and the load of each server in layer $i$ is
\[
 a_0=0,\qquad a_i=\left\lceil\frac ir\right\rceil\quad(i\geq1).
\]
First define rational layer sizes by
\begin{equation}\label{eq:layers}
 N_0=1,\qquad N_i=\frac{N_{i-1}}{a_i}\ (1\leq i<H),\qquad
 N_H=\frac{N_{H-1}}{K-1}.
\end{equation}
Multiply all sizes by
\[
 D=(K-1)((K-1)!)^r.
\]
They are then integers and the top layer contains one server. Indeed,
before scaling the top size is $1/D$, and every successive downward
ratio is an integer.

For every non-top positive layer $i$, give each server $a_i$ distinct
clients, currently assigned to it. Each client is also eligible for
one server in layer $i-1$. Pair these clients bijectively with all servers
in layer $i-1$, possible because $a_iN_i=N_{i-1}$.
The top server has $K-1$ such clients pointing down and one fixed
client eligible only for the top server. The last identity in
\eqref{eq:layers} supplies the required bijection at the top.

The current assignment has exactly the prescribed loads. Reassign all
movable clients down one layer and leave the fixed client at the top.
Every server now has exactly one client. This is a client-covering
matching $B$, and the number of clients equals the number of servers.
Every client has degree at most two. Every non-top positive-load server
has assigned load at most $K-1$ and is eligible for one client from above,
so its degree is at most $K$; the top server has degree $K$, and every
load-zero server has degree one.

The eligibility graph is a tree. Indeed, every server except the unique top
server has exactly one incoming movable client from the layer above. Following
parent links reaches the top. Excluding the fixed top client, the server-level
graph is therefore a rooted tree. The eligibility graph subdivides each of
its edges by a movable client and adds the fixed client as a leaf. This also
shows directly that the graph is connected.

Every possible nontrivial alternating move under the current assignment
goes down one layer. A path of at most $r$ moves reduces the layer index
by at most $r$, and
\[
 a_i-a_{i-t}\leq1\qquad(1\leq t\leq\min\{r,i\}).
\]
Hence the assignment is $r$-non-swappable. A path from layer $r+1$
down to layer zero has $r+1$ moves and decreases the endpoint load by
two, so the radius condition is substantive.

To evaluate its ratio, put $q_a=(a!)^{-r}$. Normalize by the layer-zero
size, and let $w_a$ be the total normalized number of load-$a$ servers.
Geometric summation of the $r$ layers at each load $1,\ldots,K-1$,
together with the single top layer, gives
\begin{equation}\label{eq:weights}
 w_0=1,\quad w_1=r,\quad
 w_a=\frac{q_{a-1}-q_a}{a-1}\ (2\leq a<K),\quad
 w_K=\frac{q_{K-1}}{K-1}.
\end{equation}
For the top load, the last recurrence in~\eqref{eq:layers} directly gives
$w_K=N_H=q_{K-1}/(K-1)$ before scaling.
Telescoping yields
\begin{align}
 \sum_{a=0}^K w_a
 &=r+2-\sum_{a=2}^{K-1}\frac{q_a}{a(a-1)}=Z_{r,K},\label{eq:mass}\\
 \sum_{a=0}^K(a-1)w_a&=0,\label{eq:balance}\\
 \sum_{a=0}^K a(a-1)w_a
 &=2+\sum_{a=2}^{K-1}q_a=2+E_{r,K}.\label{eq:excess}
\end{align}
For example, in~\eqref{eq:excess} the summands are
$a(q_{a-1}-q_a)$ for $2\leq a<K$, followed by $Kq_{K-1}$;
all intermediate coefficients cancel to one. Equation~\eqref{eq:balance}
also follows directly from the one-client-per-server comparison matching.
Since $c_a=a+a(a-1)/2$, the ratio is
\[
 \frac{\sum_a c_aw_a}{\sum_a a w_a}
 =1+\frac{2+E_{r,K}}{2Z_{r,K}}.
\]
This proves the matching lower bound and finite attainment in
Theorem~\ref{thm:main}.

\begin{corollary}[Exact gaps under degree bounds]\label{cor:degree}
For every $r\geq1$ and $K\geq2$, the exact worst locality gap over
matchable instances whose maximum server degree is at most $K$ equals
$\gap_{r,K}$. The same value holds if the maximum degree of the entire
bipartite graph is at most $K$, and in both cases it is attained on trees.
\end{corollary}
\begin{proof}
A server's assigned load cannot exceed its degree, so the finite-cap upper
bound applies. The preceding tree construction has server degree at most
$K$ and client degree at most two, which is at most $K$, and attains that
bound. Thus it supplies the matching lower bound for both graph classes.
\end{proof}

\subsection{Two concrete special cases}

For $K=2$, the sums in~\eqref{eq:sums} vanish, giving
\begin{equation}\label{eq:K2}
 \gap_{r,2}=1+\frac1{r+2}.
\end{equation}
The construction has $r+2$ clients and $r+2$ servers, with total cost $r+3$.
Its eligibility graph is a path: the top server has one fixed client and
one movable client, followed by $r$ load-one servers and one load-zero
server. Every client and server has degree at most two. The only possible
endpoint load drop of two runs from the top to the bottom and requires
$r+1$ moves. Thus the exact load-two gap, and the leading term of the
unrestricted-load answer, already occur on paths.

In particular, the worst locality gap over all matchable instances whose
eligibility graph is a path equals $1+1/(r+2)$: the preceding construction
attains it, while every server in a path has assigned load at most two.

For $r=3$, $K=3$, the complete layer data are:
\begin{center}\small
\begin{tabular}{@{}rrrrrrrrr@{}}
\toprule
Layer $i$ &0&1&2&3&4&5&6&7\\
Load $a_i$ &0&1&1&1&2&2&2&3\\
Servers &16&16&16&16&8&4&2&1\\
\bottomrule
\end{tabular}
\end{center}
There are 79 clients and 79 servers, with cost 96 and optimum 79.
The ratio is $96/79$, and the shortest improving path moves four clients.
The accompanying JSON file lists every client's eligible servers and
both assignments.

\section{Unrestricted loads and dependence on the radius}

The sums $E_r,T_r$ converge absolutely. For fixed $r$, increasing $K$
strictly increases $E_{r,K}$ and $T_{r,K}$, hence strictly increases
the expression in~\eqref{eq:finite}. Every finite assignment has a finite
maximum load. Therefore its ratio is bounded by some $\gap_{r,K}$,
whereas every such value is attained. It follows that
\[
 \gap_r=\sup_{K\geq1}\gap_{r,K}=\lim_{K\to\infty}\gap_{r,K}=\beta_r.
\]
Strict increase in $K$ shows that no finite instance attains $\beta_r$.
This completes the proof of Theorem~\ref{thm:main}.

\begin{corollary}\label{cor:radius}
The sequence $\beta_r$ is strictly decreasing. Its first value is
$\beta_1=3/2$. For every $r\geq2$,
\begin{equation}\label{eq:comparison}
 1+\frac1{r+2}<\beta_r<1+\frac1{r+1}.
\end{equation}
As $r\to\infty$,
\begin{equation}\label{eq:asymptotic}
 \beta_r=1+\frac1{r+2}+O\!\left(\frac{2^{-r}}r\right).
\end{equation}
\end{corollary}
\begin{proof}
As $r$ increases, $E_r,T_r$ strictly decrease and $r+2-T_r$ strictly
increases, proving the first claim. At $r=1$, $E_1=e-2$ and
$T_1=3-e$, giving $\beta_1=3/2$. The latter series identity follows by
summing
\[
 \frac1{a(a-1)a!}
 =\frac1{(a-1)(a-1)!}-\frac1{aa!}-\frac1{a!}
 \quad(a\geq2).
\]
The left inequality in~\eqref{eq:comparison} follows from $E_r,T_r>0$.
For the right one, note that $2T_r\leq E_r$ and
\[
 E_r\leq 2^{2-r}E_2,\qquad
 E_2\leq\frac{1/4}{1-1/9}=\frac9{32},\qquad
 (r+2)2^{2-r}\leq4\quad(r\geq2).
\]
Consequently $(r+1)E_r+2T_r\leq(r+2)E_r\leq9/8<2$, which is
equivalent to the required upper bound after multiplying positive
denominators.

Finally, factorial decay gives
$E_r=2^{-r}+O(6^{-r})$ and
$T_r=2^{-r}/2+O(6^{-r})$. Substitute these into the exact difference
\[
 \beta_r-1-\frac1{r+2}
 =\frac{(r+2)E_r+2T_r}{2(r+2)(r+2-T_r)}
\]
to obtain~\eqref{eq:asymptotic}.
\end{proof}

\begin{center}
\begin{tabular}{@{}rrrr@{}}
\toprule
Radius $r$ & Exact gap at $K=2$ & $\beta_r$ (rounded) & $1+1/r$\\
\midrule
1 & $4/3$ & 1.500000000 & 2.000000000\\
2 & $5/4$ & 1.294503159 & 1.500000000\\
3 & $6/5$ & 1.215700054 & 1.333333333\\
4 & $7/6$ & 1.172843490 & 1.250000000\\
5 & $8/7$ & 1.145423532 & 1.200000000\\
10 & $13/12$ & 1.083377417 & 1.100000000\\
\bottomrule
\end{tabular}
\end{center}
The last column gives the matchable specialization of Theorem 2 in
\cite{czygrinow}. For $r=2$, an alternative expression is
$\beta_2=1+S_0/[2(S_0+S_1)]$, where
$S_0=\sum_{a\geq0}(a!)^{-2}$ and $S_1=\sum_{a\geq1}a(a!)^{-2}$.
Theorem~\ref{thm:main} treats all radii and all finite load caps in one formula.

\section{Computational verification and code availability}

The code and generated data are publicly available at
\url{https://github.com/mgalazka84/semimatching-locality-gaps}.
The version accompanying this manuscript is commit \texttt{a04c888}.
The repository documents the model, the scope of the verification,
and commands for reproducing the reported results. Exact checks use Python
and rational arithmetic; floating-point linear-programming comparisons
are optional.

The main verification script checks 108 exact potential certificates,
covering 17,730 transition inequalities, and 108 compressed layered
constructions. It also verifies five explicit graphs, ten rational
intervals for the infinite-load constants of width below $10^{-30}$,
and 24 optional finite-state flow LPs. The explicit $(r,K)=(3,3)$ witness
has 79 servers and cost 96, as described above.

An independent script enumerates all 50,069 labelled parent maps on one
through six servers. It tests improvement by direct reachability without
using the potential or its transition rules. For $r=1,\ldots,4$, this gives
20,274 stable graph/parameter cases checked against the formula. The script
also checks twelve sharp path examples and the tree structure of the five
explicit layered examples. The enumeration exhausts the stated parent-map
family at those sizes, rather than all bipartite eligibility graphs.
These computations supplement the symbolic proof; no floating-point result
is needed for an exact bound.

\section{Concluding remarks}

The exact formula shows that excluding short improving paths gives
more precise guarantees on matchable instances than the general
$1+1/r$ bound specialized to that class. The finite-cap gap is already
attained on trees with client degree at most two, and the leading
large-radius term is attained on paths. Thus these worst local optima
do not require a dense eligibility graph.

We leave two extensions for future work. An optimum with load above
one replaces the one-parent comparison graph by a graph with larger
indegrees, so the present ancestor state no longer applies directly.
A separate algorithmic question is the distributed complexity of finding
assignments satisfying the radius condition. The exact gaps identify the
best guarantee obtainable from that condition alone on the matchable class.

\section*{Funding}

This research did not receive any specific grant from funding agencies in
the public, commercial, or not-for-profit sectors.

\section*{Declaration of competing interest}

The authors declare no competing interests.

\section*{Declaration of generative AI and AI-assisted technologies}

During the preparation of this work, the authors used ChatGPT (OpenAI) to
assist with the development and checking of verification code and with
the mathematical and editorial preparation of the \LaTeX{} manuscript.
The authors reviewed and revised the output and take full responsibility
for the results, proofs, and final text.

\end{document}